\documentclass[12pt]{amsart}

\usepackage[T1]{fontenc}
\usepackage[utf8]{inputenc}
\usepackage[english]{babel}

\usepackage{mathrsfs}

\usepackage{amsmath,amssymb,mathtools,bbm}
\usepackage{enumerate}
\usepackage{array}
\usepackage{nicefrac}

\usepackage{setspace} 
\usepackage[a4paper,tmargin=1in,bmargin=1.2in ,lmargin=1in,rmargin=1in,headheight=14.5pt,headsep=0.25in,footskip=1.5\baselineskip]{geometry}

\usepackage{natbib}
\usepackage{fancyhdr}
\usepackage{xcolor}
\usepackage{hyperref}
\hypersetup{
  breaklinks=true,
  colorlinks=true,
  linkcolor=blue,
  citecolor=blue
}
\usepackage{mleftright}
\usepackage[labelfont=md]{subcaption}
\usepackage{paralist}
\usepackage{comment}

\usepackage[normalem]{ulem}

\usepackage{amsthm}
\newtheorem{theorem}{Theorem}
\newtheorem{lemma}{Lemma}

\newtheorem{corollary}{Corollary}

\usepackage{graphicx}

\newcommand{\PP}{\mathbb{P}}
\newcommand{\EE}{\mathbb{E}}
\newcommand{\eps}{\varepsilon}
\newcommand{\abs}[1]{\left\lvert #1\right\rvert}
\newcommand{\one}{\mathbbm{1}}

\usepackage{xparse}
\DeclareDocumentCommand\Pr{ m g }{\ensuremath{
    {   \IfNoValueTF {#2}
      {\mathbb{P}\mleft[{#1}\mright]}
      {\mathbb{P}\mleft[{#1}\middle\vert{#2}\mright]}%
    }
}}
\DeclareDocumentCommand\E{ m g }{\ensuremath{
    {   \IfNoValueTF {#2}
      {\mathbb{E}\mleft[{#1}\mright]}
      {\mathbb{E}\mleft[{#1}\middle\vert{#2}\mright]}%
    }
}}

\def\ee{\mathrm{e}}

\title{On the Inefficiency of Social Learning}

\author{
Florian Brandl$^\dagger$}
\thanks{$^\dagger$Department of Economics, University of Bonn. Email: \href{mailto:florian.brandl@uni-bonn.de}{florian.brandl@uni-bonn.de}}

\author{Wanying Huang$^\ddagger$}
\thanks{$^\ddagger$Department of Economics, Monash University.
Email: \href{mailto:kate.huang@monash.edu}{kate.huang@monash.edu}}

\author{Atulya Jain$^\S$}

\thanks{$^\S$Department of Economics, University of Bonn. Email: \href{mailto:ajain@uni-bonn.de}{ajain@uni-bonn.de}}

\vspace{0.5em}
\thanks{We thank Krishna Dasaratha, Wade Hann-Caruthers, Stephan Lauermann, Omer Tamuz, and Nicolas Vieille for their helpful comments and discussions, as well as seminar and conference participants at the University of Bonn, the University of Melbourne, the University of Technology Sydney, the University of Sydney, the 37th Stony Brook International Conference on Game Theory, RINSE'25, the Paris Game Theory Seminar, and the 27th ACM Conference on Economics and Computation (EC'26).}

\date{\today}

\begin{document}

\begin{abstract}

We study whether a social planner can restore efficient learning\textemdash defined as a finite expected number of incorrect actions\textemdash in the canonical sequential social learning model. 
For each agent, the planner chooses what social information to disclose and what action-contingent transfer to offer. For unbounded signal distributions with regular tails, we show that whenever learning is inefficient without intervention, no combination of disclosure and transfers can restore efficiency with a finite budget. In particular, disclosure alone cannot restore efficiency. We identify agents' incentives as the key source of this impossibility. The budget requirement, however, concerns the transfers offered rather than the payments made. We show that even under full disclosure, the planner can achieve efficient learning using transfers while keeping  expenditure arbitrarily small. Thus, the need for an unlimited budget does not entail large payments.

\end{abstract}

\maketitle

\section{Introduction}\label{sec:intro}

A central question in social learning is: how efficiently can privately held information be aggregated in a population? The efficiency of information aggregation concerns not only whether society eventually learns the truth, but also how quickly it does so. Because individuals do not internalize how their actions affect the behavior of future agents, decentralized learning can be slow and inefficient. Such delays can be particularly costly in settings where timely decisions are crucial, such as in the adoption of new technologies and medical treatments, investment decisions, and political opinion formation before elections. This raises a natural question: can a social planner intervene to accelerate learning and restore efficiency? 

Formally, we consider a sequence of short-lived agents, each receiving a private signal about a binary state and acting once with the objective of matching the state. Unlike the standard model, in which each agent observes the actions of all her predecessors \citep*{BHW92a,Bane92a}, agents do not observe the action history directly. Instead, a benevolent social planner observes the history and decides what social information about earlier agents' actions to disclose. The planner can also offer each agent a monetary transfer to subsidize one of the two actions. This subsidy may depend on the disclosed information and is paid only if the agent takes the subsidized action. Agents are fully rational, and the planner's intervention is common knowledge. 

We measure the speed of learning using the criterion of \textit{efficient learning}, introduced by \cite{rosenberg2019efficiency}, which requires that the expected total number of agents who choose incorrectly be finite. This criterion captures not only eventual correctness of behavior\textemdash thereby ruling out incorrect herds\textemdash but also the welfare losses incurred along the path to convergence.  In the absence of the planner's intervention, \cite{smith2000pathological} show that when private signals are unbounded\textemdash so that agents can receive arbitrarily informative signals\textemdash the total number of agents who choose incorrectly must be finite almost surely. Its expectation, however, may be infinite. Indeed, as characterized by \cite{rosenberg2019efficiency}, learning is efficient only if such signals occur sufficiently frequently. This is a demanding condition that many signal distributions fail to satisfy, so learning is often inefficient.

To understand whether the planner can overcome this inefficiency, we focus on a class of unbounded signal distributions with well-behaved tails, called \textit{tail-regular} signals. Within this class, whether learning is efficient without intervention is determined by a parameter measuring tail thickness. Our main result (Theorem \ref{thm:coarse-feedback}) establishes a limit on the planner's ability to restore efficient learning. Whenever learning is inefficient without intervention, any combination of an information disclosure policy and a transfer scheme either leaves learning inefficient or requires the planner to offer substantial transfers to infinitely many agents. Thus, no planner with a finite budget can restore efficient learning. In particular, when transfers are unavailable, no fine-tuning of the disclosed social information can restore efficient learning (Corollary \ref{cor:only-feedback}).

The failure to restore efficient learning with a finite budget does not stem from a lack of information per se. If all private signals were publicly observable, learning would be efficient for any non-trivial signal distribution, including noisy binary signals that are not unbounded.\footnote{This is equivalent to a setting with a single agent who receives a private signal and takes an action in each period. In this case, the expected number of mistakes is finite, and the mistake probability in each period converges to zero exponentially fast.} 
Nor does the inefficiency arise \textit{solely} from the coarsening of private signals into binary actions.
For tail-regular signals, a planner who can dictate agents' strategies and thereby disregard their incentives could achieve efficient learning by prescribing suitable cutoff rules for the agents (Section~\ref{sec:benchmark}).
The remaining source of inefficiency is therefore agents' incentives. Because agents optimally combine their private signals with the disclosed social information and offered transfers, the planner can influence their strategies only indirectly through these instruments.

The mechanism behind Theorem \ref{thm:coarse-feedback} highlights the fundamental tradeoffs that the planner faces when trying to accelerate learning. Consider first the social information disclosure policy. The planner needs to provide precise social information to keep the current agent's mistake probability low. At the same time, he needs to extract enough information from the agent's action to inform future agents. Since the agent is rational, receiving precise social information induces her to rely heavily on it and to largely ignore her private signal. Thus, the more precise the social information, the less informative the agent's action becomes, undermining the planner's ability to learn. Conversely, imprecise social information makes the agent's action more informative, but at the cost of a higher mistake probability.

Transfers create a similar tradeoff. Intuitively, the planner should subsidize the contrarian action, that is, the action opposite to the one favored by social information. Such a subsidy would encourage agents to rely more on their private signals, thereby allowing the planner to infer more from their actions. Meanwhile, when an agent chooses the herd action despite the subsidy, her choice further strengthens the social belief that the herd action is correct. As this belief becomes more extreme, keeping future agents responsive to their private signals requires the planner to offer increasingly large subsidies. Theorem \ref{thm:coarse-feedback} shows that these tradeoffs cannot be overcome with a finite budget: any intervention that gives agents the incentives required for efficient learning either leaves the planner with too little information to sustain precise social beliefs or requires offering substantial transfers to infinitely many agents and, therefore, an infinite budget. 

An infinite-budget requirement, however, does not imply large expenditure. The planner must be prepared to offer substantial transfers repeatedly, but may rarely have to pay them. Our second result (Theorem \ref{thm:finite_exp}) makes this distinction precise: even under full information disclosure, where each agent observes the entire history of past actions, the planner can restore efficient learning using transfers while keeping his expenditure arbitrarily small. Nevertheless, by Theorem \ref{thm:coarse-feedback}, any such transfer scheme must still require an infinite budget.

Taken together, our two results identify agents' incentives as a key constraint on the planner's ability to restore efficient learning and highlight the distinction between offered transfers and actual payments. Agents' incentives depend on the transfers offered, whereas expenditure depends on which transfers are paid. The planner must therefore have unlimited capacity to offer transfers, but can design the scheme so that expenditure remains arbitrarily small.

We contribute to the literature on social learning in three ways. First, previous work has largely focused on full (social) information disclosure, showing that insufficiently informative signal structures are the key determinant of inefficient learning. By contrast, we consider general information disclosure policies and show that agents' incentives, rather than limitations of the signal structure alone, constrain the planner's ability to restore efficient learning.~%
Second, while different information structures have been shown to improve learning outcomes,\footnote{For example, see \cite{acemoglu2011bayesian, peres2020fragile}. Note that the learning outcome that they consider is asymptotic learning in probability, which is much weaker than efficient learning.} our main result shows that the scope of such improvements is limited: no information disclosure policy can restore efficient learning even when combined with a finite transfer budget. 
Third, our proofs work directly with the induced distribution of agents' beliefs, rather than analyzing trajectories of beliefs and actions across periods. This approach accommodates arbitrary disclosure policies and transfer schemes within one framework and may extend to other social learning settings or, more generally, to multi-period learning environments.

\subsection*{Related Literature}
Our paper contributes to a strand of the social learning literature that studies not only whether asymptotic learning occurs, but also the speed and efficiency of information aggregation \cite[see, e.g.,][]{vives1993fast, chamley2004rational, hann2018speed, peres2020fragile}. Most of the literature has focused on how private signal structures \citep{Bane92a, BHW92a, smith2000pathological} and observational networks \citep{ccelen2004observational, acemoglu2011bayesian, lobel2015information} determine whether agents eventually learn the state. 
By contrast, less is known about what drives fast or slow learning and whether policy interventions can accelerate this process.\footnote{For comprehensive surveys of recent developments in social learning, see \cite{golub2016learning,bikhchandani2024information}.}

A notable exception is \cite{rosenberg2019efficiency}, who provide an elegant characterization of efficient learning in a sequential social learning environment with a general signal distribution. Moreover, they show that under mild conditions on the signal distribution, this characterization is the same whether agents observe the entire history of past actions or only the action of their immediate predecessor. 
They leave open whether some social information disclosure policy can make learning efficient whenever signals are unbounded. Our main result implies that the answer is negative for tail-regular signals.

Relatedly, \cite*{xu2025social} considers a setting in which agents arrive in cohorts and observe coarse signals of past cohorts' actions, and provides necessary and sufficient conditions for learning in probability.\footnote{A similar cohort setup also appears in Section 2.6 of \cite*{rosenberg2019efficiency}.} In a complementary problem with bounded signals,
\cite*{arieli2024positive} study a regulator who can garble agents' private signals, rather than the social information, to maximize the probability that agents herd on the correct action.

Our work is also related to the recent literature on misspecification in social learning \cite[see, e.g.,][]{bohren2016informational, bohren2021learning, frick2023belief}. In particular, some forms of misspecification have been shown to improve learning. With binary signals, \cite{bernardo2001evolution} provide primarily numerical evidence that occasionally overconfident agents improve learning. In a setting more closely related to ours, \cite*{arieli2025hazards} show that, for tail-regular signals, when agents misperceive others' signals as slightly less informative than they actually are, efficient learning can be achieved in cases where learning would be inefficient under correct specification. As we discuss in Section~\ref{sec:benchmark}, the behavior induced by this misspecification can be replicated by a planner who can dictate agents' strategies, and so such a planner can always restore efficient learning. By contrast, agents in our model are fully rational and respond optimally to incentives. Our main result shows that the need to respect these incentives limits the information that a planner with a finite budget can extract from agents' actions and prevents him from restoring efficient learning.\footnote{This role of incentives is briefly discussed in  \cite{rosenberg2017efficiency}, the working-paper version of \cite{rosenberg2019efficiency}. 
There, the authors circumvent these incentive constraints by allowing a planner to dictate agents' strategies and construct a sequence of cutoff strategies that yields efficient learning for a uniform signal distribution.}

In a related setting, \cite*{smith2021informational} allow a planner to dictate agents' strategies and study the resulting problem of maximizing discounted social welfare. They show that the social optimum can be decentralized using a transfer scheme under which each agent's transfer depends on the public belief, her action, and her successor's action---a richer class of transfers than we consider. With the same objective, \cite{martynov2020essays} studies the optimal pricing policy of a planner who can subsidize and tax goods in every period. Even under these stronger intervention regimes, characterizing the social optimum remains technically challenging. We instead focus on the question of whether a social planner with a finite budget can restore efficient learning while respecting agents' incentives.

Our paper is also related to the broader literature on using information disclosure to incentivize experimentation and information acquisition.\footnote{See \citet[Chapter 11]{slivkins2019introduction} for a survey of the incentivized exploration literature.} In particular, \citet{kremer2014implementing} and \citet{che2018recommender} study bandit settings in which the planner designs recommendations to balance myopic agents' current payoffs with the informational benefits of experimentation for future agents. Closer to our setting, \citet{glazer2021wisdom} study costly acquisition of private signals in a social learning model. Our setting differs in that agents are already endowed with private information. Rather than incentivizing agents to acquire or generate information, the planner's challenge is to extract it from their actions without inducing too many mistakes. More broadly, while information disclosure can vastly improve outcomes in these settings, our results identify a fundamental limit to its effectiveness in our setting.

The way private information enters the model also distinguishes our setting from a separate strand of the literature in which agents repeatedly receive private signals and interact with each other \citep[see, e.g.,][]{molavi2018theory,harel2021rational, HST24a, brandl2025social}. In these settings, the continued inflow of private signals itself ensures eventual learning. Indeed, the mistake probability converges to zero at an exponential rate even without observing others. The question is therefore how much repeated interaction on a network can improve this rate. Existing work establishes upper bounds on the learning rate that are independent of network size when agents are myopic and interact on a complete network \citep{harel2021rational}, and when agents are forward-looking and interact on a strongly connected network \citep{HST24a}. For the slowest learner, such bounds hold even when a planner can dictate agents' strategies \citep{brandl2025social}. Our result complements these findings by identifying a related limit on the planner's ability to accelerate learning in the classic sequential social learning model, where agents are short-lived, receive only one private signal, and act only once.

\section{Model}\label{sec:model}
\subsection{Setup} There is a fixed but unknown binary state $\theta \in \{h, \ell\}$, with each state assigned prior probability $1/2$. Time is discrete with an infinite horizon: $t \in \mathbb{N} = \{1, 2, \ldots\}$. At each time $t$, an agent arrives and acts once by choosing a binary action $a_t \in \{h, \ell\}$. Each agent's base payoff is normalized to one if her action matches the state and zero otherwise.

In the baseline model, agent $t$ observes the actions of all her predecessors, $H_t = (a_1,\ldots,a_{t-1})$, before acting. She also receives a private signal $s_t \in S$, where $S$ is a measurable set of signal realizations. Conditional on $\theta$, private signals are independent and identically distributed across agents. Let $q_t= \PP[\theta= h\mid s_t]$ denote the agent's posterior belief based on the private signal. Since $q_t$ is a sufficient statistic for $\theta$ given $s_t$, we assume that $s_t = q_t$. Denote by $F_h$ and $F_\ell$ the cumulative distribution functions (CDFs) of $q_t$ conditional on $\theta= h$ and $\theta = \ell$, respectively. Let $F= \frac{1}{2}(F_h + F_\ell)$, which is the unconditional CDF of $q_t$ under the uniform prior. 

The information available to agent $t$ in the baseline model is thus $(q_t,H_t)$. She uses Bayes' rule to form her posterior belief and chooses the action that maximizes her expected payoff. That is, agent $t$ chooses $h$ if $\PP[\theta=h\mid q_t, H_t]\geq1/2$, and $\ell$ otherwise. We assume that $F_h$ and $F_\ell$ are continuous, so that agents are almost surely never indifferent between actions. 

For expositional simplicity, we further assume that the pair of conditional CDFs $(F_h, F_\ell)$ is symmetric around 1/2, i.e., $F_h(q) + F_\ell(1-q) = 1$. This implies that $F$ is also symmetric around $1/2$, i.e., $F(q) + F(1-q) = 1$. We write $\PP_h=\PP[\cdot|~\theta=h]$ for the probability measure conditional on $\theta=h$, and define $\PP_\ell$ analogously.

\subsubsection*{Tail Regularity} We impose a mild regularity condition that concerns only the tail behavior of the signal distribution. Such conditions are commonly used in the literature \citep[see, e.g.,][]{rosenberg2019efficiency, arieli2025hazards, chen2026sequential}. Specifically, the symmetric distribution $F$ is \emph{tail-regular} if there exists $\alpha >0$ such that  
\[
0 < \liminf_{q \to 0} \frac{F(q)}{q^\alpha} \leq \limsup_{q \to 0} \frac{F(q)}{q^\alpha} < \infty,
\]
in which case we write  $F(q) = \Theta(q^\alpha)$. In other words, $F(q)$ behaves like $q^\alpha$ near $q=0$. The exponent $\alpha$ thus measures the thinness of the tail: a larger $\alpha$ means that the probability of receiving a highly informative signal decays more rapidly.

For example, consider the symmetric beta family of private beliefs, whose probability density is $f_\alpha(q) = C_\alpha q^{\alpha-1}(1-q)^{\alpha-1}$, where $C_\alpha >0$ is a normalizing constant. As $q$ approaches $0$, the term $(1-q)^{\alpha-1}$ approaches one, and so the corresponding CDF $F_\alpha(q)$ behaves like $q^\alpha$ up to a multiplicative constant. Thus, the symmetric beta family is tail-regular with exponent $\alpha$.\footnote{Because beta distributions are tractable, the experimental literature often uses them to model beliefs; see, e.g., \cite{mckelvey1992experimental} and \cite*{nyarko2006informational}.} The uniform distribution of private beliefs is the special case $\alpha = 1$, for which $F(q)= q$. 

By contrast, the distribution of private beliefs induced by Gaussian signals with state-dependent means and a common variance is not tail-regular. Near $q=0$, its tail decays approximately as $\exp(-K(\log \frac{1}{q})^2)$ for some constant $K>0$, which converges to zero much faster than $q^\alpha$ for every $\alpha>0$. Consequently, $F(q)/q^\alpha \to 0$ for every $\alpha>0$, violating the tail-regularity condition.

We note that all tail-regular signals are unbounded in the sense that the support of $q_t$ contains $0$ and $1$ \citep{smith2000pathological}. Moreover, as we see below, for this class of signals the criterion for efficient learning in the baseline model reduces to a simple threshold on the tail exponent $\alpha$.

\subsubsection*{Efficiency Criterion}
Let $W = \sum_{t=1}^\infty \mathbbm{1}(a_t\neq \theta)$ be the total number of incorrect choices. Following \cite{rosenberg2019efficiency}, we say that learning is \emph{efficient} if the expected number of incorrect choices is finite, i.e.,  $\mathbb{E}[W]< \infty$, and \emph{inefficient} otherwise. 
Note that efficient learning implies that the number of incorrect choices is finite almost surely, but not vice versa. 

In the baseline model, \cite{rosenberg2019efficiency} provide the following necessary and sufficient condition for efficient learning: 
\begin{align} \label{eq:rosenberg}
    \int_0^1 \frac{1}{F(x)}\,dx < \infty.
\end{align}
Intuitively, this condition requires sufficiently informative signals to occur frequently enough so that an incorrect herd can be overturned quickly. For tail-regular signals, condition \eqref{eq:rosenberg} yields a simple threshold: learning is efficient when $\alpha<1$ and inefficient when $\alpha\geq1$.\footnote{For this class of signals, it follows from Theorem 3 of \cite{rosenberg2019efficiency} that the same efficiency criterion applies when each agent observes only her immediate predecessor's action in addition to her private signal. As we discuss below, the social planner can use information disclosure to replicate the learning outcomes under this observation structure.}

\subsection{Planner's Intervention}
We now consider a benevolent social planner who can intervene in agents' decisions to accelerate learning. His objective is to achieve efficient learning\textemdash that is, to ensure that the expected total number of incorrect actions is finite. The planner observes neither the agents' private signals nor the state $\theta$. He can influence agents' decisions through two instruments: information disclosure and transfers, which we describe in detail below. The disclosure policy and transfer scheme are common knowledge among agents.

\subsubsection*{Information Disclosure.}
Recall that in the baseline model, agent $t$'s social information consists of the full history of past actions, $H_t = (a_1, \ldots, a_{t-1})$. Under intervention, the planner instead determines what social information is disclosed to each agent. At time $t$, let $\mathscr{H}_t$ denote the information available to the planner before disclosure, including all past actions and disclosure realizations. The planner's belief is then
$$\pi_t := \mathbb{P}[\theta = h \mid \mathscr{H}_t].$$

Conditional on $\pi_t$, the planner selects, possibly at random, a signal that induces a social belief $\nu_t$ for agent $t$. 
Bayes-plausibility requires $\EE[\pi_t \mid \nu_t]=\nu_t$, and the feasible distributions of $\nu_t$ are the mean-preserving contractions of the distribution of $\pi_t$: the planner can arbitrarily coarsen the information available to him.
Without loss of generality, we take the disclosed signal to be the induced social belief $\nu_t$ itself. 
Before choosing an action, agent $t$ therefore observes the pair $(q_t,\nu_t)$:  her private signal and the social information disclosed by the planner.\footnote{For the expected number of mistakes and the expected budget and expenditure defined below, restricting the distribution of $\nu_t$ to depend only on $\pi_t$ is without loss of generality. At each time $t$, any policy that also conditions on $(\nu_1, \ldots, \nu_{t-1})$ induces a joint distribution of $(\pi_t,\nu_t)$. For a fixed transfer scheme, the conditional distribution of the planner's next belief depends on the past only through this pair. By using the original policy's conditional distribution of $\nu_t$ given $\pi_t$ at each date, the planner can replicate the joint distribution without conditioning on past induced social beliefs.}

The planner's disclosure policies range from full disclosure to no disclosure and encompass both partial and stochastic disclosure. Under \emph{full disclosure}, the planner reveals his belief, so $\nu_t=\pi_t$ for each $t$. In the absence of transfers, this policy reproduces the baseline model, in which the agent observes the full action history. Under \emph{no disclosure}, the agent receives no information about past actions, so that her social belief remains equal to the prior, $\nu_t=1/2$ for each $t$.

As an intermediate case, consider a \emph{partial disclosure} policy, under which only the $k$ most recent actions are revealed. For $t>k$, the induced social belief for agent $t$ is then $\nu_t = \PP[\theta = h \mid a_{t-k}, \dots, a_{t-1}]$. For $t\leq k$, all past actions are revealed. 
Alternatively, the planner may adopt a \emph{stochastic disclosure} policy, under which he reveals his belief $\pi_t$ with probability $\eps_t$, independently of $\pi_t$, and discloses no information otherwise.

\subsubsection*{Transfers.}
In addition to shaping agents' social beliefs, the planner can influence agents' behavior through monetary incentives. At each time $t$, after inducing social belief $\nu_t$ and before the agent acts, the planner can offer an action-contingent transfer whose amount depends on $\nu_t$. Specifically, a transfer scheme $\tau=(\tau_t)_t$ is a sequence of functions $\tau_t\colon[0,1]\to\mathbb R$. The offer to agent $t$ is determined by $\tau_t(\nu_t)$. 
Its magnitude, $|\tau_t(\nu_t)|$, determines the amount offered while its sign determines the action for which the payment is made: if $\tau_t(\nu_t)>0$, the payment is made for taking the high action; if $\tau_t(\nu_t)<0$, the payment is made for taking the low action. If $\tau_t(\nu_t)=0$ or the agent takes the other action, no payment is made. Let $\tilde{\tau}_t$ denote the transfer actually paid to agent $t$, which is equal to $|\tau_t(\nu_t)|$ if she takes the subsidized action and zero otherwise.
Thus, all payments are nonnegative, and agents are never charged. 

We restrict attention to transfer schemes that depend on the action history only through the current induced social belief. This implies that transfers do not convey additional information about the state beyond this belief. More generally, one could allow transfers to depend directly on the action history, on future agents' actions, or on the underlying state. Extending the negative result (Theorem~\ref{thm:coarse-feedback}) to a larger class of transfer schemes would strengthen it, whereas the positive result (Theorem~\ref{thm:finite_exp}) is stronger since it holds within our restricted class. We impose this restriction for simplicity, as it allows us to track the incentives offered to each agent using only the induced social belief.

We distinguish between the amount that the planner offers and the amount he actually pays. Given an information disclosure policy and a transfer scheme, define the planner's \textit{budget} as the expected total amount of transfers offered:
$$B(\tau) = \sum_{t=1}^\infty \EE\left[|\tau_t(\nu_t)|\right],$$
and his \textit{expenditure} as the expected total amount actually paid:
$$
E(\tau) = \sum_{t=1}^\infty \EE[\tilde \tau_t].$$ 
We suppress the dependence of both quantities on the disclosure policy. The planner's budget $B(\tau)$ captures the expected sum of his transfer commitments, since each offered transfer must be honored if the subsidized action is taken. In contrast, $E(\tau)$ captures the planner's expected total payment and, by construction,  satisfies $E(\tau)\leq B(\tau)$ for every transfer scheme $\tau$. Thus, a finite budget means that the planner has limited commitment capacity in expectation, whereas an infinite budget allows him to offer substantial transfers to infinitely many agents.

\subsection{Agents' Decisions Under Intervention}
Given a disclosure policy and transfer scheme, agent $t$ combines her private belief $q_t$ with the induced social belief $\nu_t$ to form her posterior belief:  $$p_t := \PP[\theta=h\mid q_t, \nu_t].$$ 
When $\tau_t(\nu_t) \geq 0$, the agent's expected payoffs from taking the high and low actions are $p_t+\tau_t(\nu_t)$ and $1-p_t$, respectively. When $\tau_t(\nu_t)<0$, they are $p_t$ and $1-p_t-\tau_t(\nu_t)$, respectively. In either case, the high action maximizes agent $t$'s expected payoff if and only if her posterior belief satisfies
\begin{equation} \label{eq:agent_decision}
   p_t \geq \frac12 (1-\tau_t(\nu_t)).
\end{equation}

In the absence of transfers, this cutoff equals $1/2$, as in the baseline model. A positive transfer, $\tau_t(\nu_t)>0$, subsidizes the high action and lowers the cutoff for taking it. Conversely, a negative transfer, $\tau_t(\nu_t)<0$, subsidizes the low action and raises this cutoff.
Throughout, we restrict attention to transfer schemes satisfying $|\tau_t(\nu)|<1$ for every $t$ and every $\nu \in(0,1)$, since a transfer of magnitude at least one makes the subsidized action optimal regardless of the agent's belief.

\section{A Benchmark} \label{sec:benchmark}
As a benchmark, we discuss the case in which the social planner can dictate agents' strategies directly. 
 This intervention is stronger than those considered above, as the planner no longer needs to respect agents' incentive constraints. We ask whether such a planner can restore efficient learning when learning is inefficient in the baseline model. 
 
The answer is positive for tail-regular signals. In the uniform case, i.e., $F(q)=q$ for all $q\in[0,1]$, \citet{rosenberg2017efficiency} construct a sequence of cutoff strategies\textemdash which differ substantially from those chosen by rational agents\textemdash that achieves efficient learning.~%
More generally, \citet{arieli2025hazards}  show that efficient learning can be achieved  for arbitrary tail-regular signals under a particular form of misspecification, in which agents are mildly ``condescending.'' 

In their setting, all agents have  tail-regular signals with exponent $\alpha$. While each agent correctly understands her own signal, she holds the misspecified belief that others' signals have exponent 
$\tilde \alpha \in (\alpha, \alpha+1)$ and are therefore less informative at the tails. Furthermore, they assume that it is common knowledge that agents hold these beliefs. This misspecification leads agents to use cutoffs that converge more slowly than the corresponding cutoffs used by correctly specified agents.\footnote{More formally, along the high path, mildly condescending agents choose the high action whenever their private belief exceeds a threshold that converges to zero at the rate $t^{-1/\tilde{\alpha}}$. The corresponding cutoff for correctly specified agents converges at the rate $t^{-1/\alpha}$. Since $\tilde{\alpha} > \alpha$, the former convergence is slower than the latter.} 
Consequently, mildly condescending agents are more likely to follow their private signals and overturn the herd. Their actions therefore reveal more private information to future agents, generating a positive informational externality sufficient to yield efficient learning.

Note that a planner who can dictate strategies can replicate the behavior of these misspecified agents by prescribing their cutoff rules. It therefore follows from Theorem 1 of \citet{arieli2025hazards} that such a planner can restore efficient learning for any tail-regular signal distribution with exponent $\alpha \ge 1$.

\section{Results}\label{sec:results}
In contrast to the benchmark, we show that once agents' incentives are taken into account, even a powerful planner who can jointly design the information disclosure policy and transfer scheme cannot restore efficient learning with a finite budget. Recall that as in the baseline model, with full disclosure and no transfers, learning is inefficient for tail-regular signals if and only if $\alpha \geq 1$. 

\begin{theorem}
\label{thm:coarse-feedback}
Suppose $\alpha \geq 1$. 
Under any information disclosure policy and transfer scheme, at least one of the following conditions must hold: 
\begin{compactenum}[(i)]
\item Inefficient learning.
\item $\limsup_{t\to \infty} |\tau_t(\nu_t)| = 1$ almost surely.
\end{compactenum}
\end{theorem}

This result states that, under any combination of a disclosure policy and a transfer scheme, either learning remains inefficient or the planner must offer transfers with magnitudes arbitrarily close to $1$ to infinitely many agents. Note that the latter condition implies that the planner's budget must be infinite, i.e., $B(\tau)= \infty$. Hence, Theorem \ref{thm:coarse-feedback} rules out efficient learning under any finite-budget scheme. In other words, efficient learning cannot be restored by a planner with limited capacity to offer transfers. 

Without transfers, an immediate consequence of Theorem~\ref{thm:coarse-feedback} is the following. 
\begin{corollary} \label{cor:only-feedback}
    Suppose $\alpha \geq 1$ and the planner imposes no transfers. Then, under any information disclosure policy, learning is inefficient.
\end{corollary}  

Thus, in our environment, the planner cannot fine-tune agents' social information to restore efficient learning through information disclosure alone.\footnote{For the class of tail-regular signals, this answers an open question raised in \cite{rosenberg2019efficiency}. They ask whether there exists an information disclosure policy under which learning is efficient when it is inefficient under full disclosure. Our result answers this question negatively.} This result does not, however, imply that full disclosure is always socially optimal. Even when learning under full disclosure is efficient, limited disclosure can sometimes improve outcomes by reducing the expected number of incorrect choices.\footnote{The idea is to delay herding and facilitate further information aggregation, possibly at the cost of a higher mistake probability for some early agents. For example, suppose agents receive one of two perfectly revealing signals (one for each state), each with probability $\eps/2$ for sufficiently small $\eps>0$, and otherwise a noisy signal inducing a posterior of $2/3$ or $1/3$ with equal probability. Under full disclosure, all agents follow agent 1 until a contradicting perfectly revealing signal arrives. If instead the planner withholds agent 1's action from agent 2, the first two actions would provide two conditionally independent observations. This reduces the probability of an incorrect herd from about $1/3$ to about $7/27$ and thus lowers the expected number of mistakes. Although this signal structure is not tail-regular, it can be approximated by a tail-regular distribution for which the same conclusion holds.} Nevertheless, when $\alpha \geq 1$, Corollary \ref{cor:only-feedback} shows that there is a limit to such improvements as they are insufficient to restore efficient learning. 

The idea behind Theorem \ref{thm:coarse-feedback} is as follows. Because agents are rational, a planner aiming to accelerate learning faces a fundamental tradeoff in designing both information disclosure and transfers. Consider first information disclosure. On the one hand, the planner must provide the current agent with sufficiently precise social information to keep her mistake probability low. On the other hand, he must extract enough information from the agent's action to generate more precise information for future agents. These two objectives conflict with each other: the more precise the agent's social information is, the more she relies on it rather than on her private signal, and the less informative her action becomes.

A similar tension arises with transfers. Intuitively, to accelerate learning, the planner should subsidize contrarian actions, making agents rely more on their private signals, so that an incorrect herd can end quickly. However, when the herd action is taken despite such a subsidy, it further strengthens future agents' belief that this action is correct and makes them more likely to follow it. To offset this effect, the planner needs to keep offering larger transfers to future agents.

In sum, Theorem \ref{thm:coarse-feedback} shows that a planner with a finite budget cannot overcome these tensions and restore efficient learning. Compared with the benchmark, the result identifies the need to respect agents' incentives as a key source of inefficiency, since a planner who can dictate agents' strategies can always achieve efficient learning.  

\subsection*{Budget vs Expenditure}
Our next result shows that the impossibility above hinges on the finite-budget constraint. In fact, we show something stronger: under full disclosure, a planner with an infinite budget can restore efficient learning through transfers alone while keeping expenditure arbitrarily small. 
\begin{theorem} \label{thm:finite_exp}
  Suppose $\alpha \geq 1$ and the planner uses full disclosure. Then, for every $\eps>0$, there exists a transfer scheme $\tau$ under which learning is efficient, 
  \[
   B(\tau)=\infty \quad \text{and} \quad E(\tau) < \eps. 
  \]
\end{theorem}
Thus, this result draws a sharp distinction between budget and expenditure. We note that the need for an infinite budget is compatible with arbitrarily small expenditure because agents' incentives depend only on the transfers offered, not on those ultimately paid.

The key ideas behind Theorem \ref{thm:finite_exp} are as follows. Consider a transfer scheme under which the planner subsidizes the contrarian action only when his belief is extreme. For example, when the planner's belief strongly favors the high action, he offers a subsidy for choosing the low action; otherwise, he offers no subsidy. The subsidy is chosen so that the agent's private-belief cutoff above which the high action is taken approaches zero more slowly than in the absence of subsidies. As a result, agents with sufficiently strong private signals in favor of the low action would continue to act against the herd and reveal their private information through their actions. This generates a positive informational externality that yields efficient learning.\footnote{This is similar to the effect of mild condescension in \cite{arieli2025hazards}.} Nonetheless, by Theorem \ref{thm:coarse-feedback}, the required budget must be infinite. 

At the same time, these subsidies are rarely paid. When the social belief strongly favors the high action, an agent chooses the low action only if she receives a sufficiently strong opposing signal\textemdash an event that occurs with small probability. The planner therefore offers subsidies only at extreme beliefs and pays them only for the unlikely contrarian action. 
Under tail-regular signals, one can construct a transfer scheme under which the subsidized region shrinks with a polynomial decay exponent strictly below one, while the probability of payment decays with an exponent strictly above one. Together, these properties ensure that the social belief converges fast enough for learning to be efficient while keeping expenditure finite. Further narrowing the region without changing its rate of shrinkage allows the planner to make his expenditure arbitrarily small.

Taken together, the two results show that restoring efficient learning requires an unlimited capacity to offer transfers. However, under full disclosure, this need not translate into large payments: efficient learning can be achieved with arbitrarily small expenditure.%

\section{Analysis}\label{sec:analysis}
In this section, we analyze the dynamics of the planner's belief and their relationship to agents' mistake probabilities. We will use these dynamics to explain the planner's tradeoff between agents' mistake probabilities, on the one hand, and the informativeness of their actions and the magnitude of the transfers offered, on the other. We conclude the section with a proof sketch of Theorem~\ref{thm:coarse-feedback}. All formal proofs are deferred to the Appendix.

\subsection{Planner's Belief Dynamics}
Recall that the planner's belief that the state is high at time $t$ is $$\pi_t = \PP[\theta = h \mid \mathscr{H}_t],
$$ 
and that agent $t$'s private belief is $q_t = \PP[\theta = h \mid s_t]$. Given $\pi_t$, the planner induces a social belief $\nu_t$ for agent $t$ through an information disclosure policy and offers a transfer $\tau_t(\nu_t)$. After the planner's intervention, agent $t$ chooses her action based on her posterior belief
\[
p_t = \PP[\theta = h \mid q_t, \nu_t].
\]
From \eqref{eq:agent_decision}, we know that agent $t$ chooses the high action if and only if
\[
p_t \ge \tfrac{1}{2}(1 - \tau_t (\nu_t)).
\]
By Bayes' rule, this condition can be written as
\[
\frac{p_t}{1-p_t} = \frac{q_t}{1-q_t} \times \frac{\nu_t}{1-\nu_t}  \geq \frac{1-\tau_t (\nu_t)}{1+\tau_t (\nu_t)}.
\]
We define agent $t$'s private-belief cutoff as 
\begin{equation} \label{eq:transfer_cut}
c(\nu_t, \tau_t (\nu_t)) 
= \frac{(1-\nu_t)(1- \tau_t (\nu_t))}{\nu_t(1+  \tau_t(\nu_t)) + (1-\nu_t)(1- \tau_t (\nu_t))}.
\end{equation} 
Under zero transfers, $c(\nu_t,0)=1-\nu_t$, so agent $t$ uses the private-belief cutoff $1-\nu_t$ when her induced social belief is $\nu_t$.
Because $\nu_t$ determines the transfer $\tau_t(\nu_t)$, we write $c_t(\nu_t)=c(\nu_t,\tau_t(\nu_t))$. 
If $\limsup_{t\to\infty}|\tau_t(\nu_t)|<1$, then $c_t(\nu_t)$ is close to $0$ whenever $\nu_t$ is close to $1$. It follows from \eqref{eq:transfer_cut} that agent $t$ chooses the high action if 
$$q_t \ge c_t(\nu_t)$$
and the low action otherwise. 

Thus, conditional on $\theta$ and $\nu_t$, agent $t$ chooses the high action with probability $1-F_\theta(c_t(\nu_t))$ and the low action with probability $F_\theta(c_t(\nu_t))$. 
After observing agent $t$ taking the high action, the planner updates his belief according to 
\begin{align}  \label{eq:high_action}
     \frac{\pi_{t+1}}{1-\pi_{t+1}} &= \frac{\pi_t}{1-\pi_t} \times \frac{1-F_h(c_t(\nu_t))}{1-F_\ell(c_t(\nu_t))}.
     \end{align}
     After observing agent $t$  taking the low action, the planner updates his belief  according to  
          \begin{align} \label{eq:low_action}
     \frac{\pi_{t+1}}{1-\pi_{t+1}}&  =  \frac{\pi_t}{1-\pi_t} \times \frac{F_h(c_t(\nu_t) )}{F_\ell(c_t(\nu_t) )}.
\end{align}
The second factor in each product is the likelihood ratio conveyed by agent $t$'s action. Because $F_h$ first-order stochastically dominates $F_\ell$, the likelihood ratio in \eqref{eq:high_action} exceeds one. So observing an agent taking the high action raises the planner's belief in the high state. Similarly, the likelihood ratio in \eqref{eq:low_action} is below one, so observing an agent taking the low action lowers that belief.

\subsection{Agent's Mistake Probability} 

We now compare the agent's mistake probability before and after she observes her private signal. 
Let $b_t$ denote the action agent $t$ would choose based only on her induced social belief and offered transfer. The next lemma shows that observing an additional private signal can reduce the agent's mistake probability by at most a factor that depends only on the private signal distribution.

\begin{lemma}\label{lem:inaccuracy-bound}
   There exists a constant $\eps > 0$, independent of the disclosure policy and transfer scheme, such that for any time $t$, 
   \[\PP[a_t \neq \theta] \geq \eps \cdot \PP[b_t \neq \theta].
   \]
\end{lemma}
The key observation behind Lemma~\ref{lem:inaccuracy-bound} is that an agent who chooses the incorrect action based on her social information and transfer alone will retain that action after any private signal that reinforces it. Under either state, such signals occur with a positive probability bounded away from zero. Therefore, conditional on her initial action being incorrect, the agent remains incorrect with probability bounded away from zero, which yields the lower bound in Lemma \ref{lem:inaccuracy-bound}.

In the absence of transfers, Lemma~\ref{lem:inaccuracy-bound} shows that the mistake probability of an agent who observes $(q_t,\nu_t)$ is at least a constant fraction of that of an agent who observes only $\nu_t$. This complements the well-known \textit{imitation principle}, which states that an agent with more information can do at least as well as one with less information \citep[see, e.g.,][]{sorensen1996rational,smith2000pathological,gale2003bayesian,golub2016learning}. Together, these bounds show that an additional private signal can reduce the agent's mistake probability only by a bounded factor. We will use this comparison to link agents' mistake probabilities to the growth rate of the planner's belief.

\subsection{Expected Belief Update}
Note first that the planner's belief $\pi_t$ is a bounded martingale and therefore converges almost surely. We denote its log-likelihood ratio at time $t$ by
$$\ell_t:=\log\frac{\pi_t}{1-\pi_t}.$$  
Taking logarithms of \eqref{eq:high_action} and \eqref{eq:low_action} gives the following updating rule:  
\begin{equation} \label{eq:LLR}
    \ell_{t+1} = \ell_t + U_t^a(\nu_t) \quad \text{if~} a_t = a \in \{h, \ell\},
\end{equation}
where we define 
\[
U_t^h(\nu) := \log \frac{1-F_h(c_t(\nu))}{1-F_\ell(c_t(\nu))}, \quad\text{and}\quad U_t^\ell(\nu) := \log \frac{F_h(c_t(\nu))}{F_\ell(c_t(\nu))}.
\]

If $F$ is tail-regular with exponent $\alpha$, then for small $q$, $F_h(q)$ and $F_\ell(q)$ are of orders $q^{\alpha+1}$ and $q^\alpha$, respectively, up to multiplicative constants (see Lemma~\ref{lem:tail_reg} in the Appendix). Thus, a private signal below a small threshold $q$ is much more likely in the low state than in the high state. An agent whose private-belief cutoff is close to zero is therefore much more likely to take the low action in the low state than in the high state.

Suppose an agent's private-belief cutoff $c_t(\nu)$ is close to zero. She then usually takes the high action, so observing her take the low action is both unexpected and highly informative to the planner. Given that the low action is more likely in the low state than in the high state, observing such an unexpected action shifts the planner's belief sharply toward the low state. In contrast, observing the agent take the high action is expected and thus less informative, so it shifts the planner's belief only slightly toward the high state. Lemma~\ref{lem:approx_llr} in the Appendix makes this asymmetry precise: when $c_t(\nu)$ is close to $0$,
\[ 
U^h_t(\nu) \approx (c_t(\nu))^\alpha \quad \text{and} \quad  U^\ell_t(\nu) \approx \log (c_t(\nu)).
\]
Thus, the likely high action raises the planner's log-likelihood ratio by only an amount of order $(c_t(\nu))^\alpha$ whereas the unlikely low action produces a large negative update of approximately $\log(c_t(\nu))$. This asymmetry in update sizes is offset by the probabilities of the corresponding events: small updates occur frequently, whereas large updates occur rarely. In consequence, even when the planner's belief converges to the truth almost surely, its log-likelihood ratio changes only gradually in expectation.

The next lemma formalizes the above argument by providing an upper bound on the expected growth of the planner's belief in log-likelihood ratio. Here, recall that $b_t$ is the action that agent $t$ would choose based on her induced social information and transfer alone, before observing her private signal.

\begin{lemma}\label{lem:bounds_mistake} Suppose $\alpha \geq 1$ and fix $0 \leq \bar\tau < 1$. Then there exists some constant $M < \infty$ such that, for every time $t$, 
    \[
\EE[\mathbbm{1}(|\tau_t(\nu_t)| \leq \bar \tau) \cdot |\ell_{t+1} - \ell_t|] \leq M \cdot \PP[b_t \neq \theta]. 
    \]
\end{lemma}
This lemma states that, under thin-tailed signals and uniformly bounded transfers, the expected absolute change in the planner's log-likelihood ratio is at most a constant multiple of the agent's mistake probability before she observes her private signal. 
Hence, even though the planner eventually learns the truth, his learning cannot accelerate too quickly.

The intuition behind Lemma~\ref{lem:bounds_mistake} is as follows. The restriction to transfers uniformly bounded away from one keeps the agents' private belief cutoffs comparable to their cutoffs under no transfers. Hence, when the induced social belief is extreme, the planner learns most when an agent acts against it. If this belief points in the correct direction, doing so requires the agent to receive an extreme private signal, which is unlikely under a thin-tailed distribution. By contrast, such a contrarian action is more likely when the induced social belief is incorrect. However, an extreme social belief is itself unlikely to point in the wrong direction. The probability that it does is precisely the mistake probability of an agent who follows her social information. Therefore, under the thin-tail condition, we can bound the planner's expected belief growth by the agent's mistake probability under social information.

\subsection{Proof Sketch of Theorem~\ref{thm:coarse-feedback}.}
Fix a tail-regular signal distribution with $\alpha\geq 1$ so that learning is inefficient under full disclosure and no transfers. Suppose, toward a contradiction, that there exists some information disclosure policy that achieves efficient learning with a transfer scheme for which $$\limsup_{t\to\infty}|\tau_t(\nu_t)| < 1$$ with positive probability. Since  $|\tau_t(\nu_t)| < 1$ for each $t$, there exists a constant $\bar\tau <1$ and an event $E$ of positive probability such that $|\tau_t(\nu_t)|\leq\bar\tau$ for every $t$ on $E$.
By Lemma~\ref{lem:bounds_mistake}, for each $t$, the expected absolute change in the planner's log-likelihood ratio, restricted to $E$, is at most a constant multiple of the agent's mistake probability before observing her private signal:
\[
\EE[\one_E|\ell_{t+1}-\ell_t|]
\leq M\PP[b_t\neq\theta].
\]
Meanwhile, Lemma \ref{lem:inaccuracy-bound} shows that the agent's mistake probability after receiving her private signal is at least a constant fraction of her mistake probability before receiving it. Hence, if learning is efficient, both these mistake probabilities are summable. Summing the bound above then gives
\[
\sup_{t \geq 1} \EE[\one_E|\ell_{t}|] \leq M \sum_{t\geq 1} \PP[b_t\neq\theta] \leq \frac{M}{\eps}  \sum_{t\geq 1} \PP[a_t\neq\theta] < \infty.
\]
Because $\PP[E]>0$, this uniform bound implies that at each $t$, the planner's belief is bounded away from zero and one with a uniformly positive probability. Hence, his mistake probability is also bounded away from zero. 
Finally, since agents' induced social beliefs are garblings of the planner's belief and transfers can only increase their mistake probabilities, Lemma~\ref{lem:inaccuracy-bound} implies that agents can outperform the planner by at most a constant factor. Hence, their mistake probabilities are likewise bounded away from zero, contradicting efficient learning.

\section{Conclusions}\label{sec:conclusion}

In this paper, we show that, for tail-regular signals with $\alpha\geq1$, a benevolent social planner with a finite budget cannot make learning efficient through any combination of information disclosure and transfers. This contrasts with the benchmark setting in which the planner can always achieve efficient learning by dictating agents' strategies. Our impossibility result underscores a key source of inefficiency in social learning that arises from the need to respect agents' incentives. We also highlight the distinction between budget and expenditure: under full disclosure, transfers alone can restore efficient learning with arbitrarily small expenditure, despite requiring an infinite budget.

We make the symmetry assumption on the private signal distribution for expositional simplicity. For example, when this assumption is relaxed to allow for two exponents $\alpha$ and $\beta$, one for each tail, the finite-budget impossibility result continues to hold if either exponent is weakly greater than one. We conjecture that our Theorem \ref{thm:coarse-feedback} extends to any signal distribution whose tails at both extremes are no thicker than those of a symmetric tail-regular distribution under which learning is inefficient with full disclosure. However, our current proof techniques apply only to tail-regular signal distributions.

 A further extension is to evaluate interventions using the discounted expected number of incorrect choices, $\sum_{t\geq1}\delta^t\PP[a_t\neq\theta]$ for some discount factor $\delta\in(0,1)$. This criterion is always finite and, therefore, allows finer comparisons across interventions. Because discounting seems to favor the full disclosure policy,\footnote{Relative to full disclosure, coarser disclosure policies shift accuracy from early to later agents. Withholding information from early agents makes their actions more informative and permits more precise disclosure later, reducing later mistakes at the cost of earlier ones. Because discounting gives greater weight to earlier mistakes, it seems to favor full disclosure.} we believe that an analogue of Theorem~\ref{thm:coarse-feedback} holds under this discounted criterion: under the same assumptions and with a fixed finite budget, planner interventions cannot substantially reduce the discounted expected number of mistakes relative to full disclosure without transfers as $\delta\to1$. This criterion is also used by \citet{smith2021informational}, who characterize the discounted social optimum when the planner can dictate agents' strategies.  

Finally, another promising direction is to study optimal disclosure policies in environments where efficient learning is feasible. Different disclosure policies may all achieve efficient learning, yet differ substantially in the expected number of mistakes.\footnote{To the best of our knowledge, characterizing the disclosure policy that minimizes the total expected number of mistakes remains an open problem.} More broadly, the limits of information disclosure as a tool for improving learning outcomes are not yet well understood. For example, it remains unclear whether full disclosure can always be improved upon, or how
to characterize the class of disclosure policies that are optimal for a given private signal distribution. It is furthermore natural to consider mechanisms that directly compensate agents for revealing their private signals or, more generally, garblings of those signals. These are interesting questions that we leave for future research.

\appendix
\section*{Appendix: Omitted Proofs from the Main Text}
\phantomsection
\label{sec:appendix}

\setcounter{equation}{0}
\renewcommand{\theequation}{A.\arabic{equation}}

 \begin{proof}[\textbf{Proof of Lemma~\ref{lem:inaccuracy-bound}}]
Recall that $q_t$ is agent $t$'s belief conditional on her private signal, and $\PP_h$ and $\PP_\ell$ are the conditional (on the state) probabilities. Define
$$\eps=\min\{\PP_\ell[q_t\geq1/2],\PP_h[q_t<1/2]\}.$$ 
Since signals are unbounded and the conditional CDFs $F_h$ and $F_\ell$ are continuous, it follows that $\eps > 0$. Write $\tau_t=\tau_t(\nu_t)$. Since $b_t$ is the action that agent $t$ would choose before receiving her private signal, by Bayes' rule, $b_t = h$ if and only if
\[
\frac{\nu_t}{1-\nu_t} \geq \frac{1-\tau_t}{1+\tau_t}.
\]
Recall that the privately informed agent chooses $a_t = h$ if and only if $q_t \geq c_t(\nu_t)$, where, by  \eqref{eq:transfer_cut}, the cutoff $c_t(\nu_t)$ satisfies
\[
\frac{c_t(\nu_t)}{1-c_t(\nu_t)} = \frac{1-\tau_t}{1+\tau_t} \times \frac{1-\nu_t}{\nu_t}.
\]    
Fix $\nu_t$ and its associated transfer $\tau_t$, and suppose that $b_t=h$. The case $b_t=\ell$ follows symmetrically. The preceding inequalities imply that
\[
c_t(\nu_t)\leq\frac12.
\]
Conditional on state $\ell$, agent $t$ therefore chooses the high action with probability
\[
\begin{aligned}
\PP_\ell[a_t=h\mid\nu_t]
&=
\PP_\ell\bigl[q_t\geq c_t(\nu_t)\bigr] \geq
\PP_\ell\left[q_t\geq\frac12\right]\geq\eps.
\end{aligned}
\]
It follows that
\[
\begin{aligned}
\PP[a_t\neq\theta\mid\nu_t]
&\geq
\PP[\theta=\ell\mid\nu_t]
\PP_\ell[a_t=h\mid\nu_t]\geq
\eps(1-\nu_t)=
\eps\PP[b_t\neq\theta\mid\nu_t],
\end{aligned}
\]
where the final equality follows because $b_t=h$. By symmetry, the same inequality holds when $b_t=\ell$. Thus, for every $\nu_t$,
\[
\PP[a_t\neq\theta\mid\nu_t]
\geq
\eps\PP[b_t\neq\theta\mid\nu_t].
\]
Taking expectations over $\nu_t$ yields
\[
\PP[a_t\neq\theta]
\geq
\varepsilon\PP[b_t\neq\theta],
\]
which proves the lemma.
\end{proof}

The following standard result \citep[see, e.g.,][]{hann2018speed,rosenberg2019efficiency,arieli2025hazards} connects the exponent $\alpha$ to the tail behavior of $F_h$ and $F_\ell$.

\begin{lemma} \label{lem:tail_reg}
    Suppose $F(q) = \Theta(q^\alpha)$. Then $F_\ell(q) = \Theta(q^\alpha)$ and $F_h(q) = \Theta(q^{\alpha+1})$.
\end{lemma}
\begin{proof}[\textbf{Proof}]
      Recall that $F$ is the unconditional cdf, while $F_\ell$ and $F_h$ are the conditional cdfs given $\theta=\ell$ and $\theta=h$. For any $q$, these conditional cdfs satisfy
\begin{equation}\label{eq:conditionalcdf}
\begin{aligned}
    F_h(q)&=2 \Big(q F(q) - \int_0^qF(x)\,dx \Big),\\
    F_\ell(q)&=2 \Big((1-q) F(q) + \int_0^qF(x)\,dx \Big).
\end{aligned}
\end{equation}
Because $|F_\ell(q)-2F(q)|\leq3qF(q)$, we have $\lim_{q\to0}F_\ell(q)/F(q)=2$. Hence, $F(q)=\Theta(q^\alpha)$ implies $F_\ell(q)=\Theta(q^\alpha)$. 

If $F(q)=\Theta(q^\alpha)$, then there exist constants $0<c\leq C<\infty$ such that $cq^\alpha\leq F(q)\leq Cq^\alpha$ for all sufficiently small $q$. Since $F_h(q)\leq2qF(q)$, it follows that $F_h(q)\leq2Cq^{\alpha+1}$.  

 For the lower bound, fix a small $q$ and some $m\in(0,1)$. We can write the integral as 
\[
\int_0^q F(x)\,dx
=\int_0^{mq} F(x)\,dx+\int_{mq}^q F(x)\,dx.
\]
Since $F$ is nondecreasing, we have $F(x)\le F(mq)$ for $x\le mq$ and
$F(x)\le F(q)$ for $x\in[mq,q]$. Hence $\int_0^{mq} F(x)\,dx \le mq\,F(mq)$, and 
$\int_{mq}^q F(x)\,dx \le (1-m)q\,F(q)$. Adding the two bounds and using the fact that $F(mq)\le C(mq)^\alpha$ gives 
\[
\int_0^q F(x)\,dx
\le Cm^{\alpha+1}q^{\alpha+1}+(1-m)q\,F(q).
\]

Using $F(q)\ge cq^\alpha$, we have
$Cm^{\alpha+1}q^{\alpha+1}\le \frac{Cm^{\alpha+1}}{c}\,qF(q)$.
Choose $m$ small enough so that $\frac{Cm^{\alpha+1}}{c}\le \frac m2$. Then
\[
\int_0^q F(x)\,dx \le \Big(1-\frac m2\Big)qF(q).
\]

Plugging into \eqref{eq:conditionalcdf} yields
\[
F_h(q)\ge 2\Big(qF(q)-\Big(1-\frac m2\Big)qF(q)\Big)=m\,qF(q)\ge mc\,q^{\alpha+1}.
\]
This provides the lower bound. Together with the upper bound, it follows that $F_h(q)=\Theta(q^{\alpha+1})$.
\end{proof}

The next lemma expresses the asymptotic size of the log-likelihood increments $U_t^h(\nu)$ and $U_t^\ell(\nu)$ in terms of the private-belief cutoff $\mu_t:=c_t(\nu)$. Here, $f(x)=o(g(x))$ means $\lim_{x\to0}f(x)/g(x)=0$, and $f(x)=O(g(x))$ means $\limsup_{x\to0}|f(x)/g(x)|<\infty$, with $g(x)\neq0$ near the limit.

\begin{lemma} \label{lem:approx_llr}
    As $\mu_t\to0$, $U_t^h(\nu)=\Theta(\mu_t^\alpha)$ and $U_t^\ell(\nu)=\log(\mu_t)+O(1)$; that is, 
    \[0 < \liminf_{\mu_t \to 0} \frac{U_t^h(\nu)}{\mu_t^\alpha} \le \limsup_{\mu_t \to 0} \frac{U_t^h(\nu)}{\mu_t^\alpha} < \infty,\] 
    and there exists a constant $C > 0$ such that for all sufficiently small $\mu_t$,
    \[
        |U_t^\ell(\nu) - \log(\mu_t)| \le C.
    \]
\end{lemma}

\begin{proof}[\textbf{Proof}]
Recall from \eqref{eq:LLR} that
\[
U_t^h(\nu) = \log \frac{1-F_h(\mu_t)}{1-F_\ell(\mu_t)}, \quad \text{and}\quad U_t^\ell(\nu) = \log \frac{F_h(\mu_t)}{F_\ell(\mu_t)}.
\]
By the assumption of tail regularity and Lemma~\ref{lem:tail_reg},  there exist constants $C_h \geq c_h>0$ and $C_\ell \geq c_\ell >0$ such that for all sufficiently small $\mu$,
$F_h(\mu)\in[c_h\mu^{\alpha+1},\,C_h\mu^{\alpha+1}]$ and
$F_\ell(\mu)\in[c_\ell\mu^\alpha,\,C_\ell\mu^\alpha]$. Substituting these into the expression for $U^\ell_t$, we obtain 
\[
\log\!\Big(\frac{c_h}{C_\ell}\mu_t\Big)
\le U_t^\ell(\nu)
\le \log\!\Big(\frac{C_h}{c_\ell}\mu_t\Big),
\]
which implies
\[
U_t^\ell(\nu)=\log(\mu_t)+O(1).
\]
From Lemma~\ref{lem:tail_reg}, we have $F_h(\mu_t)=\Theta(\mu_t^{\alpha+1})$ and $F_\ell(\mu_t)=\Theta(\mu_t^{\alpha})$, so    $F_h(\mu_t)=o(F_\ell(\mu_t))$. We rewrite $U^h_t$ as
\[
U_t^h(\nu)
=\log\!\left(\frac{1-F_h(\mu_t)}{1-F_\ell(\mu_t)}\right) \leq -\log (1-F_\ell(\mu_t)). 
\] 
Since $F_\ell(\mu_t) \leq 2 F(\mu_t)$ from Lemma~\ref{lem:tail_reg}, we have $U_t^h(\nu) \leq -\log (1-2F(\mu_t))$. Using the fact that for small $x$, $-\log(1-x) \leq x+ x^2$, it follows that
\[
U_t^h(\nu) \leq 2F(\mu_t)(1 + 2F(\mu_t)),
\]
which establishes an upper bound. For the lower bound, substitute $F_h(\mu_t)\leq2\mu_tF(\mu_t)$ and $F_\ell(\mu_t)\geq2(1-\mu_t)F(\mu_t)$ into $\ee^{U_t^h(\nu)}$:
\[
\ee^{U_t^h(\nu)} = \frac{1-F_h(\mu_t)}{1-F_{\ell}(\mu_t)} \ge \frac{1-2\mu_t F(\mu_t)}{1-2(1-\mu_t)F(\mu_t)}.
\]
Using that $1/(1-x) \ge 1+ x + x^2$ for $x\in[0,1)$ for the denominator, 
\begin{align*}
\ee^{U_t^h(\nu)} &\geq \big(1-2\mu_t F(\mu_t)\big)\big( 1 + 2(1-\mu_t)F(\mu_t) + (2(1-\mu_t)F(\mu_t))^2 \big)\\
& = 1 + (2 - 4\mu_t)F(\mu_t) + (4 - 12\mu_t + 8\mu_t^2)F(\mu_t)^2 - (8\mu_t - 16\mu_t^2 + 8\mu_t^3)F(\mu_t)^3\\
& = 1 + 2F(\mu_t) - 4\mu_t F(\mu_t) + 2(2-O(\mu_t))F(\mu_t)^2
\end{align*}
Applying the logarithm to both sides and using  $\log (1+x) \geq x -\frac{x^2}{2}$ for small $x$, we have 
\begin{align*} 
U_t^h(\nu) &\ge \log\left( 1 + 2F(\mu_t) - 4\mu_t F(\mu_t) + (4 - O(\mu_t))F(\mu_t)^2 \right) \\
&\geq 2F(\mu_t) - 4\mu_t F(\mu_t) + (2 - O(\mu_t))F(\mu_t)^2 + o(F(\mu_t)^2) \\
& \geq  2F(\mu_t) \big( 1 - 2\mu_t + (1 - O(\mu_t))F(\mu_t) + o(F(\mu_t)) \big)
\end{align*}
As $\mu_t \to 0$, we have $F(\mu_t)\to 0$, and the terms $-2\mu_t$ and $(1 - O(\mu_t))F(\mu_t)$ both vanish. The upper and lower bounds therefore imply
$\lim_{\mu_t\to0}U_t^h(\nu)/(2F(\mu_t))=1$. Since tail regularity gives $F(\mu_t)=\Theta(\mu_t^\alpha)$, we obtain $U_t^h(\nu)=\Theta(\mu_t^\alpha)$.
     \end{proof}

\begin{proof}[\textbf{Proof of Lemma~\ref{lem:bounds_mistake}}] Fix $\alpha \geq 1$ and $0\leq \bar\tau<1$. Throughout the proof, $C$ denotes a constant that may change from line to line. We also write $\tau_t = \tau_t(\nu_t)$, let $\mu_t= c_t(\nu_t)$ and define $$D_t = \{|\tau_t| \le \bar \tau\}.$$
On $D_t$, the cutoff in \eqref{eq:transfer_cut} is uniformly comparable to the no-transfer cutoff $1-\nu_t$:  
\begin{equation} \label{eq:comparability}
       \frac{1-\bar \tau}{1+\bar \tau}
        \le
        \frac{\mu_t}{1-\nu_t}
        \le
        \frac{1+\bar \tau}{1-\bar \tau}.
\end{equation}
Indeed, the denominator in \eqref{eq:transfer_cut} lies between $1-\bar\tau$ and
$1+\bar\tau$, and the same is true of the factors $1-\tau_t$ and
$1+\tau_t$. Similarly, the complementary cutoff in \eqref{eq:transfer_cut}  is uniformly comparable to the corresponding no-transfer cutoff $\nu_t$: 
\begin{equation} \label{eq:comparability2}
\frac{1-\bar \tau}{1+\bar \tau}
        \le
        \frac{1-\mu_t}{\nu_t}
        \le
        \frac{1+\bar \tau}{1-\bar \tau}.
\end{equation} 

Recall that $\mathscr{H}_t$ is the information available to the planner before deciding what to disclose at time $t$, generated by past actions and disclosure realizations. Conditional on $\mathscr{H}_t$ and the induced social belief $\nu_t$, agent $t$ chooses the low action with probability 
\[      
    \PP[a_t = \ell \mid \mathscr{H}_t, \nu_t]   =
        \pi_tF_h(\mu_t)+(1-\pi_t)F_\ell(\mu_t) =: f(\pi_t, \nu_t),
\]
since $\PP[\theta=h\mid \mathscr{H}_t,\nu_t]=\pi_t$ and, conditional on the state, the private signal $q_t$ is independent of $(\mathscr{H}_t,\nu_t)$. Thus, agent $t$ chooses the high action with probability $1-f(\pi_t,\nu_t)$, and \eqref{eq:LLR} gives
\begin{align*}
\EE\!\left[
        \abs{\ell_{t+1}-\ell_t}
        \mid \mathscr{H}_t,\nu_t
\right] =
        f(\pi_t,\nu_t)\abs{U_t^\ell(\nu_t)}
        +(1-f(\pi_t,\nu_t))\abs{U_t^h(\nu_t)}.
\end{align*}
By the law of iterated expectations,
\begin{align}
\EE\!\left[
        \abs{\ell_{t+1}-\ell_t}
        \mid \nu_t
\right]
&= \EE\left[\EE\!\left[
        \abs{\ell_{t+1}-\ell_t}
        \mid \mathscr{H}_t, \nu_t
\right] \mid \nu_t\right] \notag \\
&= \EE\left[ 
        f(\pi_t,\nu_t)\abs{U_t^\ell(\nu_t)}
        +(1-f(\pi_t,\nu_t))\abs{U_t^h(\nu_t)}
\mid \nu_t\right] \notag \\
&= 
        \big(\nu_tF_h(\mu_t)+(1-\nu_t)F_\ell(\mu_t)\big)
        \abs{U_t^\ell(\nu_t)}
        \notag \\
&\qquad\qquad
        +
        \big(\nu_t(1-F_h(\mu_t))+(1-\nu_t)(1-F_\ell(\mu_t))\big)
        \abs{U_t^h(\nu_t)},
\label{eq:cond-increment}
\end{align}
where the last equality follows since $\EE[\pi_t \mid \nu_t] = \nu_t$. We will show that the right-hand side is bounded by
$C \min\{\nu_t,1-\nu_t\}$.

First consider $\mu_t\leq\eps$, where $\eps>0$ is small enough for Lemma~\ref{lem:approx_llr} to apply. On $D_t$, \eqref{eq:comparability} gives $1-\nu_t\leq C\mu_t$ and $\mu_t\leq C(1-\nu_t)$. When $\mu_t\leq\eps$, Lemmas~\ref{lem:tail_reg} and~\ref{lem:approx_llr} give $\abs{U_t^\ell(\nu_t)}\leq C\abs{\log\mu_t}$, $F_h(\mu_t)=\Theta(\mu_t^{\alpha+1})$, and $F_\ell(\mu_t)=\Theta(\mu_t^\alpha)$. 
Thus, on $D_t$ the first term in \eqref{eq:cond-increment} is bounded by
\begin{align*}
 C\big(\mu_t^{\alpha+1}+(1-\nu_t)\mu_t^\alpha\big)
        \abs{\log\mu_t} 
\le
        C\mu_t^{\alpha+1}\abs{\log\mu_t}
        \le C\mu_t,
\end{align*}
where the last inequality follows since $x^\alpha |\log x| \to 0$ as $x \to0$ for any $\alpha>0$.  
Lemma~\ref{lem:approx_llr} gives $\abs{U_t^h(\nu_t)}\leq C\mu_t^\alpha$. Since $\alpha\geq1$, the second term in \eqref{eq:cond-increment} is therefore at most $C\mu_t$. For sufficiently small $\eps$, $\nu_t\geq1/2$, so $1-\nu_t=\min\{\nu_t,1-\nu_t\}$. Therefore,
$$ \EE\left[ \mathbbm{1}_{D_t} \cdot  \abs{\ell_{t+1}-\ell_t} \mid \nu_t \right] \leq C\mu_t \leq C(1-\nu_t) = C\min\{\nu_t, 1-\nu_t\}. $$
The case $\mu_t\ge1-\varepsilon$ follows from a symmetric argument.

In the middle region $\mu_t\in(\varepsilon,1-\varepsilon)$, both $\abs{U_t^h(\nu_t)}$ and $\abs{U_t^\ell(\nu_t)}$ are bounded. On $D_t$, the two bounds in \eqref{eq:comparability} and \eqref{eq:comparability2} imply that both  $\nu_t$ and $1-\nu_t$ are bounded away from zero. Hence, \eqref{eq:cond-increment} is again bounded by
$C\min\{\nu_t,1-\nu_t\}$.

Combining the three cases, 
\begin{align*}
        \EE\!\left[\mathbbm{1}_{D_t} \cdot 
        \abs{\ell_{t+1}-\ell_t}
        \mid\nu_t
        \right]
        \le
        C  \min\{\nu_t,1-\nu_t\}.
\end{align*}
Recall that $b_t$ is the action chosen by an agent with social belief $\nu_t$ under transfer $\tau_t$ before observing her private signal. Let $b_t^0$ be the action chosen by such an agent with zero transfer. Since the agent chooses the action that minimizes the probability of a mistake, her conditional mistake probability is 
$\PP[b_t^0\neq\theta \mid \nu_t] = \min\{\nu_t,1-\nu_t\}$. Because the zero-transfer action minimizes the mistake probability, any nonzero transfer can only weakly increase it. Hence, 
\[
        \min\{\nu_t,1-\nu_t\} \le  \PP[b_t\neq\theta \mid \nu_t].
\]
Taking the expectation over $\nu_t$ gives the desired inequality: 
$$ \EE[\mathbbm{1}_{D_t} \cdot  |\ell_{t+1}-\ell_t|] \le C \EE[\min\{\nu_t, 1-\nu_t\}] \le C \PP[b_t \ne \theta]. $$
\end{proof}

\begin{proof}[\textbf{Proof of Theorem~\ref{thm:coarse-feedback}}]

Fix $\alpha\geq1$. Denote by $a^x_t$ the action chosen by the planner to match the state based on his belief $\pi_t$ at time $t$. Let $a^0_t$ be the action chosen by the agent with posterior belief $p_t$ under zero transfer. Notice that from her perspective, her mistake probability is $\PP[a_t^0 \neq \theta \mid q_t, \nu_t]=\min\{p_t, 1-p_t\}$. Since any nonzero transfer can only weakly increase the probability of choosing the wrong action,
\[
\PP[a^0_t \neq \theta \mid  q_t , \nu_t] \leq \PP[a_t \neq \theta \mid  q_t, \nu_t].
\]
Thus, $\EE[\min\{p_t,1-p_t\}] = \PP[a^0_t \neq \theta] \leq \PP[a_t \neq \theta]$.
By Lemma~\ref{lem:inaccuracy-bound}, there exists $\eps>0$ such that for all $t$, 
$$\PP[a^0_t \neq \theta] \geq \eps \cdot \PP[b^0_t \neq \theta],$$
where $b^0_t$ is the action chosen by an agent with induced social belief $\nu_t$ under zero transfer. 
Because $\EE[\pi_t\mid\nu_t]=\nu_t$ and $x\mapsto\min\{x,1-x\}$ is concave, Jensen's inequality gives $$
\PP[b^0_t \neq \theta] = \EE[\min\{\nu_t, 1-\nu_t\}] \geq \EE[\min\{\pi_t, 1-\pi_t\}] =\PP[a^x_t \neq \theta].$$
Combining the above inequalities yields
\begin{align} \label{eq:garble}
    \begin{aligned}
        \PP[a_t \neq \theta] & \ge \eps\cdot\PP[a^x_t \neq \theta].
    \end{aligned}
\end{align}

Now suppose, toward a contradiction, that learning is efficient but 
\[
\PP\Big[\limsup_{t \to \infty} |\tau_t(\nu_t)|<1\Big] >0. 
\]
It follows from efficient learning and \eqref{eq:garble} that $\sum_{t\geq1}\PP[a_t^x\neq\theta]<\infty$. Since $\pi_t=\frac{\ee^{\ell_t}}{\ee^{\ell_t}+1}$, the planner's mistake probability is 
$$\PP[a^x_t \neq \theta] = \EE[\min\{\pi_t, 1-\pi_t\}] =  \EE\Big[\frac{1}{\ee^{|\ell_t|}+1}\Big],$$
which is also summable. This implies that almost surely,
$$\sum_{t\geq1}\frac{1}{\ee^{|\ell_t|}+1}<\infty,$$ which implies that $\abs{\ell_t}\to\infty$ almost surely. 

Recall that $|\tau_t| < 1$ for every $t$. By the hypothesis, there exist $0\leq\bar\tau<1$ and $T<\infty$ such that 
\[
        E=
        \{\abs{\tau_t}\le\bar\tau\text{ for all }t\ge T\}
\]
occurs with positive probability. For every $t>T$, the triangle inequality gives $\abs{\ell_t}
        \le
        \abs{\ell_T}
        +
        \sum_{s=T}^{t-1}\abs{\ell_{s+1}-\ell_s}$.
Since $\abs{\ell_T}<\infty$ almost surely, there is
$K<\infty$ such that the event $   E'=E\cap\{\abs{\ell_T}\le K\}$
also occurs with positive probability. Taking expectations yields that for every $t > T$,
\begin{align*}
        \EE[\one_{E'}\abs{\ell_t}]
        \le
        K+
        \sum_{s=T}^{t-1}
        \EE[\one_{E'}\abs{\ell_{s+1}-\ell_s}] & \le
        K+
        \sum_{s=T}^{t-1}
        \EE\!\left[
        \one(\abs{\tau_s}\le\bar\tau) \cdot 
        \abs{\ell_{s+1}-\ell_s}
        \right] \\
        &\le
        K+M\sum_{s=T}^{\infty}\PP[b_s\neq\theta],
\end{align*}
where the last inequality follows from Lemma~\ref{lem:bounds_mistake}.

By Lemma~\ref{lem:inaccuracy-bound} again, there exists $\delta>0$ such that for every $t\geq1$, $\delta\PP[b_t\neq\theta]\leq\PP[a_t\neq\theta]$. It then follows from efficient learning that $\sum_{t\geq1}\PP[b_t\neq\theta]<\infty$, which implies, from  the preceding inequality, that  $$\sup_{t>T}\EE[\one_{E'}\abs{\ell_t}]<\infty.$$ 

But efficient learning implies that  $\abs{\ell_t}\to\infty$ almost surely, so  $\one_{E'}\abs{\ell_t}\to\infty$ on $E'$. Since $\PP[E']>0$, Fatou's lemma gives
\[
        \infty
        =
        \EE\!\left[\liminf_{t\to\infty}\one_{E'}\abs{\ell_t}\right]
        \le
        \liminf_{t\to\infty}\EE[\one_{E'}\abs{\ell_t}],
\]
reaching a contradiction to the uniform bound above. This concludes the proof.

\end{proof}

The following lemmas will be useful in proving Theorem \ref{thm:finite_exp}. The first shows that when the private-belief cutoff lies in $[c,1-c]$, the information conveyed by each action is at least of order $c^\alpha$. The second uses this bound to show that when $c_t$ decreases sufficiently slowly, the probability that the planner's belief is not sufficiently close to the true state is summable.

	\begin{lemma} \label{lemma:log-likelihood-lowerbound}
		Suppose \(F(q)=\Theta(q^\alpha)\). Then there exists \(\kappa>0\) such that, for all
		sufficiently small \(c>0\) and all \(\mu\in[c,1-c]\),
		\[
		\sqrt{F_h(\mu)F_\ell(\mu)}
		+
		\sqrt{(1-F_h(\mu))(1-F_\ell(\mu))}
		\leq
		1-\kappa c^\alpha.
		\]
	\end{lemma}
	\begin{proof}[\textbf{Proof}]
    For a given cutoff $\mu$, let
		\[
		A(\mu)
		=
		\sqrt{F_h(\mu)F_\ell(\mu)}
		+
		\sqrt{(1-F_h(\mu))(1-F_\ell(\mu))}.
		\]
		We use the identity
 		\[
		1-A(\mu)
		=
		\frac12\left(\sqrt{F_\ell(\mu)}-\sqrt{F_h(\mu)}\right)^2
		+
		\frac12\left(\sqrt{1-F_h(\mu)}-\sqrt{1-F_\ell(\mu)}\right)^2.
		\]
		Since both terms are nonnegative, it suffices in each region of $\mu$ to lower bound either one of them by a constant multiple of $c^\alpha$.
		
	By Lemma \ref{lem:tail_reg}, we have \(F_h(q)/F_\ell(q)\to0\) as \(q\to0\). Choose \(\eps\in (0, 1/2)\) such that
		for every \(q\leq\varepsilon\),
		\[
		F_h(q)\leq \frac14F_\ell(q).
		\]
		
		Suppose first that \(\mu\in[c,\varepsilon]\). It follows from the preceding inequality that $\sqrt{F_h(\mu)}\leq \frac12\sqrt{F_\ell(\mu)}$. Thus, we have
		\begin{align*}
		1-A(\mu)
		& \geq
		\frac12\left(\sqrt{F_\ell(\mu)}-\sqrt{F_h(\mu)}\right)^2 \\
        & 
		\geq
		\frac18F_\ell(\mu)\\
        &\geq K \mu^\alpha \geq K c^\alpha,
		\end{align*}
        where the third inequality follows from $F_\ell(\mu) = \Theta(\mu^\alpha)$.

		Next, suppose that  \(\mu\in[1-\varepsilon,1-c]\). Let $q = 1-\mu$, so that $q \in [c, \eps]$. By symmetry, 
		\[
		1-F_h(\mu)=F_\ell(q),
		\qquad
		1-F_\ell(\mu)=F_h(q).
		\]
        The same argument applied to the second squared term therefore gives 
        \[
        1-A(\mu)\geq \frac12\left(\sqrt{1-F_\ell(\mu)}-\sqrt{1-F_h(\mu)}\right)^2 = \frac12\left(\sqrt{F_\ell(q)}-\sqrt{F_h(q)}\right)^2 \geq K c ^\alpha.
        \]
		
		Finally, suppose that  $\mu\in[\varepsilon,1-\varepsilon]$. Because $F_h(\mu) < F_\ell(\mu)$ for every $\mu \in (0,1)$, and since 
        $F_h(\mu)/F_\ell(\mu)$ is continuous on the compact interval $[\eps, 1-\eps]$,  there exists $\rho <1$  such that $F_h(\mu)\leq \rho F_\ell(\mu)$
        throughout this region. Hence,
		\[
		1-A(\mu)
		\geq
		\frac12\left(\sqrt{F_\ell(\mu)}-\sqrt{F_h(\mu)}\right)^2
		\geq
		\frac12(1-\sqrt{\rho})^2F_\ell(\mu) \geq \frac12(1-\sqrt{\rho})^2F_\ell(\eps) >0.
		\]
        For all sufficiently small $c$, this uniform positive lower bound also satisfies $$\frac12(1-\sqrt{\rho})^2F_\ell(\eps) \geq K c^\alpha.$$
        Taking $\kappa>0$ to be the minimum of the constants obtained in the three regions yields
        \[
        1-A(\mu) \geq \kappa c^\alpha 
        \]
        for every $\mu \in [c, 1-c]$, which proves the lemma. 
	\end{proof}

    \begin{lemma}\label{lem:belief-concentration}
Suppose $F(q)=\Theta(q^\alpha)$. Fix $\beta>\alpha$ and $\eta\in(0,1/2)$, and let
\[
    c_t=\eta t^{-1/\beta}.
\]
Suppose that the planner uses full disclosure and that the private-belief cutoff satisfies
\[
    \mu_t\in[c_t,1-c_t]
\]
at every history. Then
\[
    \sum_{t\geq1}\PP_h(\pi_t<1-c_t)<\infty
    \qquad\text{and}\qquad
    \sum_{t\geq1}\PP_\ell(\pi_t>c_t)<\infty.
\]
\end{lemma}

\begin{proof}[\textbf{Proof}]
We prove the result under state $h$. The result under state $\ell$ follows by symmetry.
    Recall that the log-likelihood ratio of the planner's belief is
$\ell_t=\log\frac{\pi_t}{1-\pi_t}$ and define $B_t = \log \frac{1-c_t}{c_t}$. By the likelihood-ratio updating rule in \eqref{eq:LLR}, we can write  
\[
 e^{-\frac12(\ell_{t+1}-\ell_t)} = \begin{cases}
\sqrt{\frac{1-F_\ell(\mu_t)}{1-F_h(\mu_t)}},
& \text{if }a_t=h,\\[9pt]
\sqrt{\frac{F_\ell(\mu_t)}{F_h(\mu_t)}},
& \text{if }a_t=\ell.
\end{cases}
\]
Conditional on $H_t$ and  $\theta = h$, the probabilities of taking action $\ell$ and action $h$ are $F_h(\mu_t)$ and $1-F_h(\mu_t)$, respectively. Therefore, 
\[
\EE_h\left[
    e^{-\frac12(\ell_{t+1}-\ell_t)}
    \mathrel{\Big|} H_t
\right]
=\sqrt{F_h(\mu_t)F_\ell(\mu_t)}
  +\sqrt{\bigl(1-F_h(\mu_t)\bigr)
          \bigl(1-F_\ell(\mu_t)\bigr)}.
\]
Since $\mu_t\in[c_t,1-c_t]$ and $c_t \to 0$, Lemma~\ref{lemma:log-likelihood-lowerbound} implies that there exist $\kappa>0$ and $T< \infty$ such that, for all $t\geq T$,
\[
\begin{aligned}
\EE_h\left[
    e^{-\frac12(\ell_{t+1}-\ell_t)}
    \mathrel{\Big|} H_t
\right]
&\leq 1-\kappa (c_t)^\alpha.
\end{aligned}
\]
Define $Z_t=e^{-\ell_t/2}$. Since $Z_t$ is $H_t$-measurable, for every $t \geq T$, 
\[
\begin{aligned}
\EE_h[Z_{t+1}\mid H_t]
&=Z_t\EE_h\left[
    e^{-\frac12(\ell_{t+1}-\ell_t)}
    \mathrel{\Big|}H_t
\right]\leq (1-\kappa (c_t)^\alpha)Z_t.
\end{aligned}
\]
Taking expectations and iterating from $T$ gives, for every \(t>T\),
\[
\begin{aligned}
\EE_h[Z_t]
&\leq C\prod_{s=T}^{t-1}(1-\kappa (c_s)^\alpha) \leq C\exp\left(-\kappa\sum_{s=T}^{t-1}(c_s)^\alpha\right),
\end{aligned}
\]
where the second inequality uses $1-x\leq e^{-x}$.

Since $\{\pi_t<1-c_t\}=\{\ell_t<B_t\}=\{Z_t>e^{-B_t/2}\}$, by Markov's inequality 
\[
\PP_h(\pi_t<1-c_t)
\leq e^{B_t/2}\EE_h[Z_t]
\leq C\exp\left(
    \frac{B_t}{2}
    -\kappa\sum_{s=T}^{t-1}(c_s)^\alpha
\right).
\]

Let $p=\alpha/\beta<1$. Since $(c_s)^\alpha=\eta^\alpha s^{-p}$, an integral comparison gives
\[
    \sum_{s=T}^{t-1}s^{-p}
    \geq \int_T^t x^{-p}\,dx
    =\frac{t^{1-p}-T^{1-p}}{1-p}.
\]
Consequently, for some $K_\eta>0$ and all sufficiently large $t$,
\[
    \sum_{s=T}^{t-1}(c_s)^\alpha
    \geq K_\eta t^{1-\alpha/\beta}.
\]
Moreover,
\[
    B_t
    =\log\frac{1-c_t}{c_t}
    \leq -\log c_t
    =-\log\eta+\frac{1}{\beta}\log t.
\]
It follows that, for some constants $C,K_\eta>0$,
\[
    \PP_h(\pi_t<1-c_t)
    \leq C\exp\left(
        -K_\eta t^{1-\alpha/\beta}+C\log t
    \right).
\]
Because $1-\alpha/\beta>0$, the power term $t^{1-\alpha/\beta}$ eventually dominates $\log t$ and so the right-hand side is summable. Hence,
\[
    \sum_{t\geq1}\PP_h(\pi_t<1-c_t)<\infty.
\]
\end{proof}

Now we are ready to prove Theorem \ref{thm:finite_exp}. Building on the preceding lemmas, we construct a transfer scheme under which the information conveyed by agents' actions is at least of order $c_t^\alpha$, whereas the probability that a subsidy is paid is of order $c_t^{\alpha+1}$. Choosing $c_t=\eta t^{-1/\beta}$ with $\beta\in(\alpha,\alpha+1)$ allows information to accumulate sufficiently quickly for efficient learning while keeping expenditure finite. 

\begin{proof}[\textbf{Proof of Theorem \ref{thm:finite_exp}}]
Fix $\alpha\geq 1$ and suppose that the planner uses full disclosure, so that $\nu_t=\pi_t$ for every $t$. We construct a transfer scheme under which learning is efficient and $E(\tau)<\eps$. Theorem~\ref{thm:coarse-feedback} will then imply that $B(\tau)=\infty$.   

Fix $\beta\in(\alpha,\alpha+1)$ and, for some $\eta \in (0, 1/2)$ to be chosen below, define
\[
    c_t=\eta t^{-1/\beta}.
\] 
The planner chooses transfers to induce the following private-belief cutoff:
\[
\mu_t(\pi_t)=
\begin{cases}
\max\{c_t,1-\pi_t\}, & \text{if}~ \pi_t\ge \frac12,\\[3pt]
\min\{1-c_t,1-\pi_t\}, & \text{if}~ \pi_t<\frac12.
\end{cases}
\]
Thus, $\mu_t\in[c_t,1-c_t]$ for every $t\geq1$. 
Writing $\mu_t=\mu_t(\pi_t)$ and $\tau_t = \tau_t(\pi_t)$, the transfer that implements this cutoff is 
$$ \tau_t = \frac{1 - \mu_t - \pi_t}{(1 - \mu_t)(1 - \pi_t) + \mu_t\pi_t}. $$

The planner offers a transfer only when his belief is extreme, and the transfer subsidizes the contrarian action. If $\pi_t > 1-c_t$, then $\mu_t = c_t > 1-\pi_t$, so $\tau_t(\pi_t)<0$ and the planner subsidizes action $\ell$. If $\pi_t< c_t$, then $\mu_t = 1- c_t < 1-\pi_t$, so $\tau_t(\pi_t)>0$ and the planner subsidizes action $h$. If $\pi_t \in [c_t, 1-c_t]$, then $\mu_t=1-\pi_t$ and hence $\tau_t(\pi_t)=\tilde\tau_t=0$. Throughout the remainder of the proof, $C$ denotes a positive constant that may change from line to line and is independent of $t$ and $\eta$, provided that $\eta$ is sufficiently small.

We first show that the expenditure can be made arbitrarily small. Fix a history $H_t$. If $\pi_t > 1-c_t$, action $\ell$ is subsidized and $\mu_t = c_t$. Since a payment is made only when the subsidized action is chosen, $\tilde\tau_t\leq\mathbbm{1}(a_t=\ell)|\tau_t|\leq\mathbbm{1}(a_t=\ell)$. 
It follows that
\[
\begin{aligned}
\mathbb E[\tilde\tau_t \mid H_t]
&\le \mathbb P(a_t = \ell \mid H_t) \\
&= \pi_t F_h(c_t) + (1-\pi_t)F_\ell(c_t) \\
&\le Cc_t^{\alpha+1}+C(1-\pi_t)c_t^\alpha \\
&\le C\cdot c_t^{\alpha+1},
\end{aligned}
\]
where the penultimate inequality follows from Lemma \ref{lem:tail_reg}, and the final inequality uses $1-\pi_t<c_t$. 

If $\pi_t < c_t$, then action $h$ is subsidized and $\mu_t = 1-c_t$. Analogously, $\tilde\tau_t\leq\mathbbm{1}(a_t=h)|\tau_t|\leq\mathbbm{1}(a_t=h)$ and hence
\[
\begin{aligned}
\mathbb E[\tilde\tau_t \mid H_t]
&\le \mathbb P(a_t = h \mid H_t) \\
&= \pi_t\bigl(1-F_h(1-c_t)\bigr) + (1-\pi_t)\bigl(1-F_\ell(1-c_t)\bigr) \\
&= \pi_t F_\ell(c_t) + (1-\pi_t)F_h(c_t) \\
&\le C \cdot c_t^{\alpha+1},
\end{aligned}
\]
where the second equality follows from the symmetry of $(F_h,F_\ell)$ around $1/2$, and the inequality follows from Lemma~\ref{lem:tail_reg} and $\pi_t<c_t$.

When $\pi_t\in[c_t,1-c_t]$, no payment is made. Therefore, the bound
\[
    \EE[\tilde\tau_t\mid H_t]\leq Cc_t^{\alpha+1}
\]
holds at every history. By iterated expectations,
\begin{align*}
E(\tau)= \sum_{t\geq 1}  \EE[\tilde \tau_t] = \sum_{t\geq 1}\EE\bigl[\EE[\tilde\tau_t\mid H_t]\bigr] \le C\cdot \sum_{t\geq 1} c_t^{\alpha+1} = C \eta^{\alpha+1}\sum_{t\ge1}t^{-\frac{\alpha+1}{\beta}}.
\end{align*}
Because  $\beta<\alpha+1$, the series on the right-hand side is finite. Thus, by choosing $\eta>0$ sufficiently small, we obtain $E(\tau)<\eps$. 

We next show that the same transfer scheme achieves efficient learning. Because $\beta>\alpha$ and $\mu_t\in[c_t,1-c_t]$ at every history,
Lemma~\ref{lem:belief-concentration} gives
\[
    \sum_{t\geq1}\PP_h(\pi_t<1-c_t)<\infty
    \qquad\text{and}\qquad
    \sum_{t\geq1}\PP_\ell(\pi_t>c_t)<\infty.
\]

Suppose first that $\theta=h$. Then
\begin{align*}
    \PP_h(a_t=\ell)
    &\leq \PP_h(\pi_t < 1-c_t) + \PP_h(a_t=\ell, \pi_t \ge 1-c_t) \\
    &= \PP_h(\pi_t < 1-c_t)
    + \EE_h\left[\one_{\{\pi_t \ge 1-c_t\}}\PP_h(a_t=\ell \mid H_t)\right] \\
    &= \PP_h(\pi_t < 1-c_t)+\PP_h(\pi_t \ge 1-c_t) F_h(c_t) \\
    &\leq \PP_h(\pi_t < 1-c_t) + F_h(c_t),
\end{align*}
where the second equality uses the fact that $\mu_t=c_t$ on
$\{\pi_t\geq1-c_t\}$. The first term is summable by Lemma \ref{lem:belief-concentration}.  The second is also summable because
Lemma~\ref{lem:tail_reg} gives  $F_h(c_t)\leq Cc_t^{\alpha+1}$ and $\beta<\alpha+1$. Hence,
\[
    \sum_{t\geq1}\PP_h(a_t=\ell)<\infty.
\]
By symmetry,
\[
    \sum_{t\geq1}\PP_\ell(a_t=h)<\infty.
\]
Thus, learning is efficient.

Finally, because the transfer scheme achieves efficient learning and $\alpha\geq1$, Theorem~\ref{thm:coarse-feedback} implies
\[
    \limsup_{t\to\infty}|\tau_t(\nu_t)|=1
    \qquad\text{almost surely}.
\]
Thus, $|\tau_t(\nu_t)| \geq 1/2$ for infinitely many $t$ almost surely, so  $\sum_{t\geq1}|\tau_t(\nu_t)|=\infty $ almost surely. Hence $B(\tau) =  \EE\big[\sum_{t\geq 1}|\tau_t(\nu_t)|\big]= \infty$, which concludes the proof.

\end{proof}

\end{document}